\documentclass[
aps,pra,
superscriptaddress,
longbibliography,
a4paper,
reprint
]{revtex4-2} 
\usepackage[utf8]{inputenc}
\usepackage[english]{babel}
\usepackage[T1]{fontenc}

\usepackage{graphicx}
\usepackage{mathtools}
\usepackage{braket}
\usepackage{amsmath}
\usepackage{amssymb}
\usepackage{amsthm}
\usepackage{units}

\usepackage{xcolor}
\usepackage{thm-restate}
\usepackage{dcolumn}
\usepackage{multirow}

\usepackage{hyperref}
\hypersetup{colorlinks=true}
\hypersetup{citecolor=blue}
\hypersetup{linkcolor=blue}
\hypersetup{urlcolor=blue}

\newtheorem{lem}{Lemma}
\newtheorem{res}{Result}
\newcommand{\vl}[1]{\vec{\lambda}_{#1}}
\DeclareMathOperator\tr{Tr}
\newcommand*\dd{\mathrm{d}}
\DeclarePairedDelimiter\abs\lvert\rvert
\DeclarePairedDelimiter\norm\lVert\rVert

\DeclareRobustCommand{\stirling}{\genfrac\{\}{0pt}{}}

\begin{document}

\title{Bounding the classical cost of simulating quantum behaviors in the  prepare-and-measure scenario}

\author{Sebastian Schl\"osser}
\email{sebastian.schlosser@physics.uu.se}
\affiliation{Department of Physics and Astronomy, Uppsala University, Box 516, 75120 Uppsala, Sweden}
\affiliation{Nordita, KTH Royal Institute of Technology and Stockholm University, 10691, Stockholm, Sweden}
\affiliation{Naturwissenschaftlich-Technische Fakultät, Universität Siegen, 57068 Siegen, Germany}
\author{Matthias Kleinmann}
\email{matthias.kleinmann@uni-muenster.de}
\affiliation{Department for Quantum Technology, University of Münster, 48149 Münster, Germany}
\affiliation{Naturwissenschaftlich-Technische Fakultät, Universität Siegen, 57068 Siegen, Germany}


\begin{abstract}
We study the prepare-and-measure scenario in which Alice transmits a quantum system to Bob, who then performs a quantum measurement. The quantum state of the system is unknown to Bob, and the measurement is unknown to Alice. It has recently been shown that shared randomness and two bits of classical communication are necessary and sufficient to simulate the transmission of a qubit. We show that the communication cost can be reduced to an average of \unit[1.89]{bits}. We then study restricted sets of state preparations: First, for a restriction to real-valued qubit states, if the communication of a classical trit is sufficient, we show that the corresponding protocol must have a convoluted form. We then reduce the smallest qubit scenario requiring two bits of classical communication to only 6 state preparations and 5 measurements. For a qutrit, it is not known whether the communication cost is finite; we identify a scenario that requires at least $5$ classical messages, already for the simulation of the real qutrit. Finally, we develop a method for restricted sets of states, that allows us to lower bound the classical communication cost based solely on the set of quantum states.
\end{abstract}

\maketitle

\section{\label{sec:intro}Introduction}
A fundamental question in quantum information theory regards the advantage of quantum systems over their classical counterparts in information processing tasks. Holevo showed that quantum systems and their classical counterparts are equivalent for transmitting classical information \cite{Holevo}. This result was generalized by Frenkel and Weiner \cite{Frenkel_2015}, showing that classical systems can reproduce the correlations of quantum systems of equal size, provided that the receiver has a fixed measurement. Consequentially, incompatible measurements are necessary for a quantum advantage in the prepare-and-measure scenario \cite{Vieira_2023,de_Gois_2021,Egelhaaf}.
However, for communication tasks where the output also depends on an input given to Bob, quantum systems can provide an advantage over their classical counterparts. One such example is random access coding, where the transmission of a qubit instead of a classical bit can improve the success probability \cite{Wiesner1983ConjugateC,ambainis1998densequantumcodinglower,ambainis2009quantumrandomaccesscodes}. In the general prepare-and-measure scenario, one investigates classical models that simulates the quantum statistics, in particular using only finite classical communication \cite{Brassard_exponential_gap,Cerf_2000,Massar_ClSimEnt,STEINER_2000,TonerBacon,Renner_PandM,Degorre}. A breakthrough was made by Toner and Bacon \cite{TonerBacon}, as they showed that a qubit prepare-and-measure scenario restricted to projective measurements can be simulated classically with two bits of communication and shared randomness. This result was recently extended to generalized measurements (POVMs) by Renner et al.\ \cite{Renner_PandM}: A qubit prepare-and-measure scenario can be simulated classically by two bits of communication and the protocol was shown to be optimal in terms of the worst-case communication cost. To the best of our knowledge, it is an open problem whether quantum systems beyond qubits can be simulated using finite classical communication, even when considering only the average of the message length.

In this work, we show that the two bits of communication in the qubit simulation protocols by Toner and Bacon \cite{TonerBacon} and Renner et al.\ \cite{Renner_PandM} can be compressed, resulting in a single-shot protocol for simulating a qubit with an average communication cost of \unit[1.89]{bits} and a worst-case communication cost of \unit[3]{bits}. We discuss the simulation of the real qubit, where the results of Gallego et al.\ \cite{Gallego_2010} and Renner et al.\ \cite{Renner_PandM} imply that an optimal simulation protocol for the real qubit requires either a trit ($3$ messages) or two bits ($4$ messages) of communication. We show that if the real qubit admits a simulation using a classical trit, the encoding step in a corresponding protocol is necessarily convoluted. We then turn to minimal simulation scenarios and demonstrate that a separation between qubit correlations and classical correlations using a classical trit already occurs in a scenario with only $6$ state preparations and $5$ projective measurements, improving on the smallest previously known construction. Regarding the open problem of the qutrit, we show that the simulation of the real qutrit requires strictly more communication than the qubit case, that is, we show that an alphabet consisting of at least $5$ messages is required. Finally, we prove a structural property of all classical simulation models by relating the problem to quantum state discrimination. In particular, any simulation model is restricted by the distinguishability of the quantum states in a corresponding state discrimination task. We provide a necessary condition for the existence of a classical simulation protocol, which can be expressed as a linear program. The program depends only on a set of input states and provides a particularly efficient method to certify lower bounds on the worst-case communication cost.

\section{\label{sec:cl_sim_PM}Classical simulation of the prepare-and-measure scenario}
We consider the prepare-and-measure scenario, where Alice prepares a quantum state and sends it to Bob who performs a quantum measurement on the state. We denote the dimension of the quantum system by $d_Q$. A quantum state is described by a density operator, that is, a positive-semidefinite operator $\rho$ with $\tr(\rho)=1$. A quantum measurement is described by a positive operator-valued measure (POVM), that is, a collection of positive-semidefinite operators $\set{M_b}_b$ with $\sum_b M_b = \openone$.

\begin{figure}
    \centering
    \includegraphics[width=0.7\linewidth]{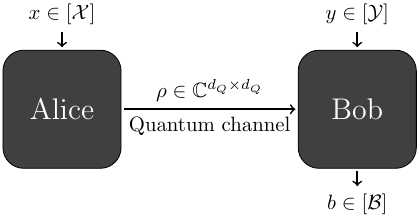}
    \caption{The quantum prepare-and-measure scenario where Alice receives an input $x$, prepares a $d_Q$-dimensional quantum state $\rho$ and transmits it to Bob. Bob receives an input $y$ and performs a quantum measurement which produces an outcome $b$.}
    \label{fig:PM_quantum}
\end{figure}

The referee provides Alice with a label $x$ and asks her to prepare the quantum state $\rho_x$ and transmit it to Bob. Bob is provided with a label $y$ that specifies the measurement and performs the measurement represented by the POVM $\set{M_{b\vert y}}_{b}$, where $b$ denotes the measurement outcome. The resulting outcome probabilities are given by Born's rule, yielding the behaviors
\begin{equation}
    p_Q(b\vert x,y) = \tr(\rho_x M_{b\vert y}).
\end{equation}
For settings with a finite number of inputs and outputs, we denote the cardinality of Alice's input set by $\mathcal{X}$, the cardinality of Bob's input set by $\mathcal{Y}$ and the cardinality of Bob's output set by $\mathcal{B}$. This scenario is illustrated in Fig.~\ref{fig:PM_quantum}.

We consider classical models for this task where we allow Alice and Bob to have access to a shared random variable $\lambda$ with range $\Lambda$ and distribution $\pi(\lambda)\dd\lambda$. That is, in each round, Alice and Bob have access to the same random value $\lambda$, but this value is not known to the referee. As before, the referee provides Alice with the label $x$, but in the classical model she can only transmit a classical system $c\in[d_C]\coloneqq\set{1,\ldots,d_C}$ of ``dimension'' $d_C$ to Bob. Alice selects $c$ with  probability $p_A(c\vert x,\lambda)$ and submits it to Bob. Bob is provided with the label $y$ and asked to provide an outcome $b$. Correspondingly, he provides the outcome $b$ with probability $p_B(b|c,y,\lambda)$. This scenario is illustrated in Fig.~\ref{fig:PM_classical}. Therefore, the behaviors for the classical model are given by
\begin{equation}
\label{eq:cl_model}
    p_C(b\vert x,y) = \int_\Lambda \dd\lambda\,\pi(\lambda) \sum_{c=1}^{d_C} p_A(c\vert x,\lambda) p_B(b\vert c,y,\lambda).
\end{equation}
\begin{figure}
    \centering
    \includegraphics[width=0.7\linewidth]{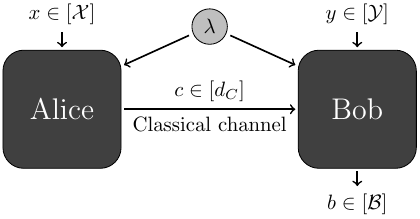}
    \caption{The classical prepare-and-measure scenario where Alice and Bob share classical correlations $\lambda$. Alice receives an input $x$, prepares a classical message from an alphabet of length $d_C$ and transmits it to Bob. Bob receives an input $y$ and provides an outcome $b$.}
    \label{fig:PM_classical}
\end{figure}%
The $d_Q$-dimensional quantum prepare-and-measure scenario can be simulated by a $d_C$-dimensional classical system if there exist encodings $p_A(c\vert x,\lambda)$, decodings $p_B(b\vert c,y,\lambda)$ and a suitable $\pi(\lambda)$ on $\Lambda$ such that
\begin{equation}
    p_C(b\vert x, y) = p_Q(b\vert x,y),
\end{equation}
for all possible inputs $x$ and $y$, and outputs $b$.

We mention that without loss of generality, one can assume that Alice acts deterministically by absorbing her local randomness into the shared randomness. To this end, we introduce a new variable $\lambda^\prime = (\lambda,\mu)$, where $\mu$ is from the uniform distribution over $[0,1]$.
We can now write the response function of Alice as
\begin{equation}
    p_A(c\vert x,\lambda) = \int_0^1 \text{d}\mu\, D_A(c\vert x,\lambda^\prime),
    \label{eq:prob_to_detstr_1}
\end{equation}
where
\begin{equation}
    D_A(c\vert x,\lambda^\prime) =
    \begin{cases}
        1 \text{ if } \mu \le p_A(c\vert x,\lambda)\\
        0 \text{ else,}
    \end{cases}
    \label{eq:prob_to_detstr_2}
\end{equation}
is a deterministic response function. The same argument can be applied for Bob's response function.
Expressing the classical model in terms of deterministic response functions for Alice and Bob, it allows for an implementation as a linear program \cite{Gallego_2010,Renner_PandM}, which we discuss in Appendix \ref{appendix:A}. However, as the number of deterministic strategies (and therefore the number of variables and constraints in the linear program) grows exponentially in the number of inputs and outputs, the program is only tractable for small instances. In Appendix \ref{appendix:A}, we describe a method to reduce the number of deterministic strategies in the linear program, allowing us to tackle larger instances of the problem. 

Previous works have considered two main ways to  quantify the classical communication cost of a simulation protocol. In the first approach one uses the length of the classical alphabet $d_C < \infty$ to quantify the cost of a simulation protocol \cite{TonerBacon,Renner_PandM}. We refer to this metric as the worst-case communication cost $\mathcal{C} = \log_2 d_C$.
In the second approach, instead of requiring the communication to be finite in the worst case, one requires the communication to be only finite on average \cite{Hansen_SimpleAlgo,Maudlin_1992,STEINER_2000,Cerf_2000}. In other words, we allow for an infinite alphabet length $d_C$ but we require the \emph{average communication cost}
\begin{equation}
  \bar{\mathcal C}=
   \sup_x\int \dd\lambda\, \pi(\lambda) l(x,\lambda)
\end{equation}
of the message to be finite. Here, $l(x,\lambda)$ denotes the number of bits of the classical message that is assigned to the input $x$ for a given $\lambda$. Clearly, every protocol with finite $\mathcal{C}$ corresponds to a protocol with finite $\bar{\mathcal{C}}$. Thus, we have that for a given setting,
\begin{equation}
    \min\bar{\mathcal{C}} \le \min\mathcal{C},
\end{equation}
where each minimization is performed over all simulation protocols.  It was shown that without shared randomness, a classical simulation of a qubit with finite worst-case communication cost $\mathcal{C}$ is impossible \cite{Massar_ClSimEnt}. This directly implies that a protocol with a finite amount of shared randomness, $\abs{\Lambda}<\infty$, is impossible, because any finite randomness could be included into finite communication. Note that in a scenario where one allows for two-way communication, every quantum system can be simulated with finite average communication cost \cite{Massar_ClSimEnt}. 

Toner and Bacon \cite{TonerBacon} established $\min \mathcal{C} \le \unit[2]{bits}$, for all qubit states and projective measurements. Renner et al.\ \cite{Renner_PandM} recently showed $\min \mathcal{C}= \unit[2]{bits}$, for all qubits states and generalized measurements (POVMs), thus solving the most general qubit case for this notion of communication cost.

It has been shown that a classical simulation requires a classical alphabet length $d_C$ that scales at least exponentially in the quantum dimension $d_Q$ \cite{Buhrmann_exponential_gap,Brassard_exponential_gap,Montina2011,Havlicek2020}. The approach by Montina in Ref.~\cite{Montina2011} can be formulated as a state discrimination task of orthogonal quantum states, and the proof uses a result regarding the maximal size of distance evading sets on hyperspheres. The authors of Ref.~\cite{Havlicek2020} consider a task involving quantum state antidistinguishability and in their proof they use spherical codes. Apart from their general exponential lower bounds, they show that a simulation protocol for a qutrit ($d_Q=3$) has $d_C\ge5$ and a simulation for a ququat ($d_Q=4$) has $d_C\ge 10$.

However, beyond the qubit case, no upper bounds on the communication cost are known. In particular, it is not known whether the qutrit admits a classical simulation with finite communication, that is, with $d_C<\infty$.

\section{\label{sec:sim_protocols}Simulation protocols}
\subsection{Qubit}
\label{subsec:sim_qubit}
Renner et al.\ showed that in terms of the worst-case communication cost $\mathcal{C}$, two bits of communication are necessary and sufficient to simulate a qubit \cite{Renner_PandM}. Furthermore, they also addressed possible protocols with an average communication cost $\bar{\mathcal{C}}$ lower than \unit[2]{bits}. They showed that when Alice sometimes sends a bit (or less), a simulation of a qubit is impossible. Their proof only applies to the case where the length of the message that Alice sends is decided only by the shared randomness. However, in general, the length of the message can in addition depend on the input state.

We show now that in the general case, a qubit can be classically simulated with an average communication cost of $\bar{\mathcal{C}} = \unit[1.89]{bits}$. To this end, we employ a variable-length encoding \cite{huffman,yin2021multichannelhuffmancodeasymmetricalphabet} and thus, in some cases Alice sends just $1$ bit, while in other cases, she sends up to \unit[3]{bits}. This demonstrates that a simulation of a qubit with a nonzero bit part is possible. Our protocol is readily constructed from the protocol by Renner et al.\ \cite{Renner_PandM}, the only difference lies in the encoding and decoding of the classical message. It is sufficient to consider pure input states $\rho_x$, as the existence of a classical model for a set of pure states implies the existence of a protocol for their convex hull.

We recall the first two steps of the protocol \cite{Renner_PandM} which also resemble the first steps in the protocol by Toner and Bacon \cite{TonerBacon}. Here, $\Theta(z)=1$ when $z\ge 0$ and $\Theta(z)=0$ when $z<0$.
\begin{itemize}
            \item[1.] Alice and Bob share two vectors $\vl{1},\vl{2} \in \mathbb{R}^{3}$, which are uniformly and independently distributed on the unit radius sphere $\mathbb{S}^{2}$.
            \item[2.] Alice is given a pure state $\rho=\frac{1}{2}(\openone+\vec{x}\cdot{\vec{\sigma}})$. She prepares two bits via the formula $c_1 = \Theta(\vec{x}\cdot\vl{1})$ and $c_2 = \Theta(\vec{x}\cdot\vl{2})$ and sends them to Bob.
\end{itemize}
The two transmitted bits $c_1,c_2$ do not necessarily contain two bits of information, as they are in general correlated for an observer that has access to the shared randomness $\vl{1},\vl{2}$ \cite{TonerBacon}. In fact, the only case where $c_1$ and $c_2$ are not correlated is when $\vl{1}$ and $\vl{2}$ are orthogonal. On the other hand, when $\vl{1}$ and $\vl{2}$ are (anti)-parallel, the two bits are perfectly (anti)-correlated and contain only one bit of information.

We define four regions on the Bloch sphere,
\begin{equation}\label{eq:sphere_regions}
            S_{i,j}(\vl{1},\vl{2}) = \{\vec{r}\in\mathbb{S}^2 \mid \Theta(\vec{r}\cdot\vec{\lambda}_{1})=i, \Theta(\vec{r}\cdot\vec{\lambda}_{2})=j\},
\end{equation}
where $i,j\in\set{0,1}$. They correspond to the sets of states which are assigned the same message. In the original protocols, when $\vec{x}\in S_{i,j}(\vl{1},\vl{2})$, Alice transmits the bits $c_1=i$ and $c_2=j$. The probability $p(i,j|\alpha,x)$ for the message $(i,j)$ becomes independent of the input $x$, when integrated over the shared randomness with fixed angle $\alpha = \measuredangle(\vl{1},\vl{2})\in[0,\pi]$, since then only the relative angles between $\vec x$, $\vl{1}$, and $\vl{2}$ occur in the integral. We obtain
\begin{equation}\begin{split}\label{eq:Halpha}
            &p(0,0|\alpha)=p(1,1|\alpha)=\frac{1}{2}-\frac{\alpha}{2\pi},\\
            &p(0,1|\alpha)=p(1,0|\alpha)=\frac{\alpha}{2\pi}.
\end{split}\end{equation}
Shannon's source coding theorem \cite{shannon} establishes that the entropy of a distribution forms a lower bound for the best possible lossless compression rate. In the limit of large coding block length, the optimal compression rate converges to the entropy. The entropy for the above probability distribution is given by
\begin{equation}
    H(\alpha) = -2\left[\frac{\alpha}{2\pi}\log_2 \frac{\alpha}{2\pi}+\left(\frac{1}{2}-\frac{\alpha}{2\pi}\right)\log_2\left(\frac{1}{2}-\frac{\alpha}{2\pi}\right)\right].
\end{equation}
Thus, using an encoding that only depends on $\alpha$, the lower bound for compression is given by
\begin{equation}
\int_0^\pi \text{d}\alpha \frac{\sin{\alpha}}{2} H(\alpha)\approx \unit[1.85]{bits},
\end{equation}
see also Ref.~\cite{TonerBacon}. Here, $\frac{1}{2}\sin\alpha=\frac1{4\pi}(\int \dd\phi) \sin \alpha$ is the distribution of the enclosed angle $\alpha$ of two vectors that are drawn independently from a uniform distribution on the unit sphere. 

Importantly, the communication can be compressed in the single-shot case, that is, without performing multiple rounds of the simulation in parallel. 

\begin{res}
    A qubit admits a single-shot classical simulation with an average communication cost of $\bar{\mathcal{C}} = \unit[1.89]{bits}$ and a worst-case communication cost of $\mathcal{C} = \unit[3]{bits}$.
\end{res}

To this end, we employ Huffman coding \cite{huffman,yin2021multichannelhuffmancodeasymmetricalphabet} with blocks of size $1$, which provides a minimal encoding. Using binary Huffman coding, we find a binary encoding A$^\pm$ that, on average, requires less than two bits of communication for $\alpha\le\frac{\pi}{3}$ and $\alpha\ge\frac{2\pi}{3}$. Ternary Huffman coding leads to the mixed-channel encoding B$^\pm$ where one always sends a trit, and sometimes an additional bit. The Huffman codes of higher order are trivial, since the original code has four symbols. Therefore, we define our encoding piecewise on $\alpha$ by choosing the optimal encoding for each $\alpha$, see also Fig.~\ref{fig:entropy_curve}. Then, the first two steps of the protocol read as follows.
\begin{itemize}
            \item[1. ] Alice and Bob share two vectors $\vec{\lambda}_{1},\vec{\lambda}_{2} \in \mathbb{R}^{3}$, which are uniformly and independently distributed on the unit radius sphere $\mathbb{S}^{2}$, and they define the four regions $S_{i,j}(\vec{\lambda_1},\vec{\lambda_2})$ on the unit sphere, see Eq.~\eqref{eq:sphere_regions}.        
            \item[2. ] Alice is given a pure state $\rho=\frac{1}{2}(\openone+\vec{x}\cdot{\vec{\sigma}})$. She informs Bob in which of the four regions $S_{i,j}(\vec{\lambda}_1,\vec{\lambda}_2)$ the state $\vec{x}$ lies in, that is, she sends $(i,j)$ to Bob, and encodes her message depending on $\alpha=\pi \log_2\beta$, where $\alpha = \measuredangle(\vl{1},\vl{2})$, as follows:
\end{itemize}
\begin{equation}
\begin{tabular}{l|l|l|l|l|l}
Encoding & \multicolumn{1}{c|}{A$^+$} & \multicolumn{1}{c|}{B$^+$} & \multicolumn{1}{c|}{C} & \multicolumn{1}{c|}{B$^-$} & \multicolumn{1}{c}{A$^-$}\\[0.5ex]
$\beta$ & $\left[1,\frac{9}{8}\right]$ & $\left(\frac{9}{8},\frac{4}{3}\right]$ & $\left(\frac{4}{3},\frac{3}{2}\right)$ & $\left[\frac{3}{2},\frac{16}{9}\right)$ & $\left[\frac{16}{9},2\right]$ \\[1ex] \hline
        $S_{00}$ & $0_2$   & $0_3$  & $00_2$ & $2_31_2$ & $111_2$ \\
        $S_{11}$ & $10_2$  & $1_3$  & $11_2$ & $2_30_2$ & $110_2$ \\
        $S_{01}$ & $110_2$ & $2_30_2$ & $01_2$ & $1_3$  & $10_2$  \\
        $S_{10}$ & $111_2$ & $2_31_2$ & $10_2$ & $0_3$  & $0_2$
        \end{tabular}        
\end{equation}

\begin{figure}
    \centering
    \includegraphics[width=\linewidth]{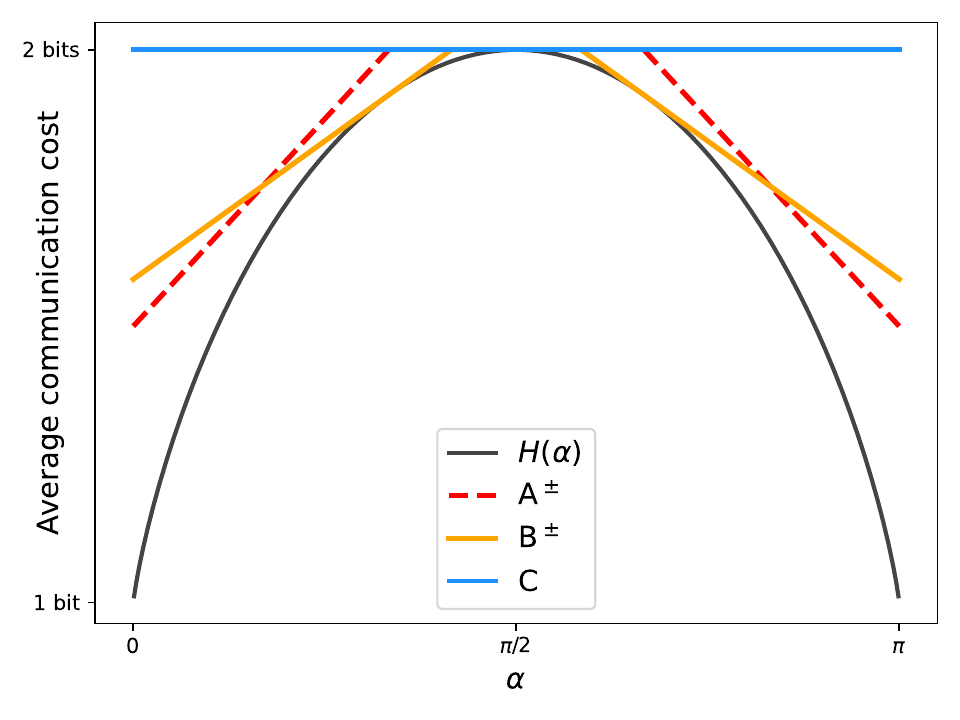}\vspace{-15pt}
    \caption{Average communication cost for different encoding strategies at a given angle $\alpha= \measuredangle(\vl{1},\vl{2})\in[0,\pi]$. The $\alpha$-asymptotic encoding yields a lower bound, given by the entropy $H(\alpha)$, see Eq.~\eqref{eq:Halpha}. For single-shot encoding, the average communication cost of $5$ different strategies ($\mathrm A^\pm$, $\mathrm B^{\pm}$, $\mathrm C$) is shown, where encoding $C$ is the original encoding. Choosing the encoding with the lowest cost at each $\alpha$ yields our piecewise single-shot encoding with average communication cost of $\bar{\mathcal C}\approx \unit[1.89]{bits}$.}
    \label{fig:entropy_curve}
\end{figure}

The encodings (and decodings) depend on the shared randomness $\lambda$ through $\alpha$. Furthermore, as we employed a variable-length encoding, the length of the message is determined by the shared randomness $\lambda$ and by the state $\vec{x}$. Note that when $\beta \le \frac{9}{8}$ or $\beta \ge \frac{16}{9}$, in some cases Alice sends just one bit, while in other cases she sends up to three bits. When $\frac{9}{8}<\beta\le\frac{4}{3}$ or $\frac{3}{2}\le\beta<\frac{16}{9}$, Alice always sends a trit and sometimes an additional bit. For $\frac{4}{3}<\beta<\frac{3}{2}$, Alice performs the original encoding and always sends two bits. The average message length of the different encodings for a given angle $\alpha$ is plotted in Fig.~\ref{fig:entropy_curve}.

The average communication cost of this protocol is given by
\begin{equation}\begin{split}
    \bar{\mathcal{C}} &=
    \sup_{\vec x} \int \dd\lambda_1\dd\lambda_2 \pi(\vec\lambda_1,\vec \lambda_2) l(\vec x,\vec\lambda_1,\vec \lambda_2)\\
    &=\int_0^\pi\text{d}\alpha \frac{\sin\alpha}{2} L(\alpha) \approx \unit[1.89]{bits},
\end{split}\end{equation}
where $\frac{\sin\alpha}{2} = \frac{1}{4\pi}\int_0^{2\pi} \dd\phi \sin\alpha$ represents the probability density of the angle $\alpha$ enclosed by two independent uniform vectors on the unit sphere. $L(\alpha)$ denotes the average codeword length for the encoding described above. That is, $L(\alpha)=\sum_{i,j} p(i,j|\alpha)L_{i,j}(\alpha)$, where $L_{i,j}(\alpha)$ is the length of the encoding for $S_{i,j}$. We evaluate the integral in Appendix~\ref{appendix:B}.

\subsection{Real qubit}
\label{subsec:sim_real_qubit}

In this section we will consider the \emph{real qubit} or \emph{rebit}, which corresponds to the $xz$-plane of a qubit. That is, the qubit states and measurement effects are restricted to lie within the $xz$-plane and can can therefore be described with real numbers. Since the real qubit is contained in a qubit, it is clear that $d_C=4$ is sufficient but is not known whether it is also necessary for a classical protocol. In Ref.~\cite{Gallego_2010} a setting of real qubit states and measurements was provided that does not admit a classical simulation with $d_C=2$ and thus, for the real qubit, $\min d_C \in \set{3,4}$. We show that if $\min d_C = 3$, then, in the corresponding simulation protocol, Alice's message has a convoluted dependence on her input. 
\begin{figure}
    \centering
    \includegraphics[width=\linewidth]{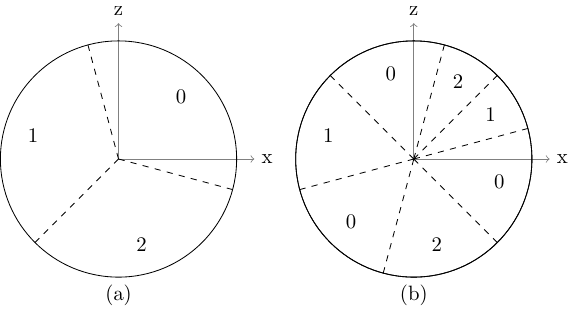}
    \caption{A deterministic assignment $D_A(c\vert x,\lambda)$ for pure states of the real qubit with $d_C=3$ and (a) $N(\lambda) = 3$ sections and (b) $N(\lambda) = 7$ sections.}
    \label{fig:real_qubit_sections}
\end{figure}

Given a deterministic strategy $D_A(c\vert x,\lambda)$, we define $N(\lambda)$ as the number of times that the classical message $c$ changes when traversing all states along the equator. This definition is illustrated in Fig.~\ref{fig:real_qubit_sections}. Clearly, for an optimal protocol with a classical alphabet of length $d_C$, it holds that
\begin{equation}
    \max_\lambda N(\lambda) \ge d_C.
\end{equation}
If the above equation does not hold, the corresponding protocol is not optimal since it always assigns less than $d_C$ messages.
On the other hand, $N(\lambda)$ can be arbitrarily large which consequently makes the formulation of the corresponding protocol convoluted. We note that the qubit protocols by Toner and Bacon \cite{TonerBacon} and Renner et al.\ \cite{Renner_PandM} with $d_C=4$ require no more sections than necessary, that is, $\max_\lambda N(\lambda) = 4$.

\begin{res}
    For an optimal classical simulation protocol for the real qubit, either (i) $d_C = 3$ and $\max_\lambda N(\lambda) \ge 7$ or (ii) $d_C = 4$
    holds.
\end{res}

It remains to show that a hypothetical simulation protocol for the real qubit with $d_C=3$ has $\max_\lambda N(\lambda) \ge 7$. We use a finite number of states and measurements to obtain bounds on $N(\lambda)$, since the full simulation of the real qubit must in particular be able to simulate any finite subset. For a given $\lambda$, the minimal number of sections of an assignment for all states that is compatible with $D_A(c\vert x,\lambda)$ is lower bounded by
\begin{equation}
    N(\lambda)\ge \abs{\set{ x\in\mathcal X | c(x,\lambda)\ne c(x+1,\lambda)}},
\end{equation}
where $c(x,\lambda)\in[d_C]$ is the message sent for $x$ and $\lambda$, and addition is taken modulo $\mathcal{X}$.

We consider the $(\mathcal{X},\mathcal{Y},\mathcal{B}) = (32,16,2)$ setting. Alice receives input $x\in[\mathcal{X}]$, and prepares the state $\rho_x = \frac{1}{2}(\openone + \cos(2\pi x/\mathcal{X})\sigma_x+\sin(2\pi x/\mathcal{X})\sigma_z)$. Bob receives input $y\in[\mathcal{Y}]$ and performs the binary measurement $M_{\pm\vert y} = \frac{1}{2}(\openone \pm \cos(\pi y/\mathcal{Y})\sigma_x\pm\sin(\pi y/\mathcal{Y})\sigma_z)$. The corresponding linear program (see Appendix~\ref{appendix:A}), where we exclude all deterministic strategies $D_A(c\vert x,\lambda)$ for which $N(\lambda) \ge 7$, is not feasible. We conjecture that $N(\lambda)\ge 7$ is not fundamental and can be pushed further with more computational resources.

\section{\label{sec:lower_bounds} Classical dimension witnesses}
In Ref.~\cite{Gallego_2010}, a concept for classicality certification of behaviors in the prepare-and-measure scenario was introduced. A general classical dimension witness $W$ can be written as
\begin{equation}
    W=\sum_{b,x,y} c_{x,y}^b\,  p_C(b\vert x,y) \le C_{d_C},
\end{equation}
with real coefficients $c_{x,y}^b$. All classical models with an alphabet of length $d_C$ produce behaviors $\{p_C(b\vert x,y)\}_{b,x,y}$ that respect the classical bound $C_{d_C}$. In contrast, a violation of the classical bound $C_{d_C}$ certifies that no classical model using an alphabet of size $d_C$ can simulate the given behavior.

For a setting with binary outcomes $b\in\set{\pm}$, a classical dimension witness is conveniently expressed in terms of expectation values, that is,
\begin{equation}
\label{eq:witness_correlators}
    W = \sum_{x,y} c_{x,y} E_{x,y} \le C_{d_C}^\prime,
\end{equation}
with coefficients $c_{x,y}$ and $E_{x,y} = p(+\vert x,y) - p(-\vert x,y)$.

Firstly, we will introduce witnesses that show a separation between qubit behaviors and all behaviors which are trit-simulable. Secondly, we will provide a witness that shows a separation between behaviors that can be obtained from a real qutrit and all behaviors which are classically simulable with $d_C=4$. This proves our result that a simulation of the real qutrit has $d_C\ge5$.

\subsection{Qubit}
Gallego et al.\ showed that there exist qubit behaviors which cannot be simulated by a classical bit \cite{Gallego_2010}. They characterized the classical polytope for $d_C=2$ and the $(\mathcal{X},\mathcal{Y},\mathcal{B})=(3,2,2)$ setting and showed that only one non-trivial facet exists. The corresponding inequality is a tight witness for $d_C=2$ which can be violated by (real) qubits. Renner et al.\ showed that in the $(6,11,2)$ setting, there exist qubit behaviors which are not trit-simulable, proving optimality of their simulation protocol~\cite{Renner_PandM}.

We aim to go towards a minimal example. First, when $d_C\ge\mathcal{X}$, a perfect encoding is possible and the simulation becomes trivial. Furthermore, when Bob has only one input, a perfect classical simulation is possible with $d_C=d_Q$ \cite{Frenkel_2015}. Therefore, any scenario that features correlations which are not trit-simulable has $\mathcal{X}\ge4,\mathcal{Y}\ge2$ and $\mathcal{B}\ge2$.
The smallest setting that we identify here is $(6,5,2)$. For this, we consider a setting with six states $\vec{x}_i$ given by the Bloch vectors $\pm\hat{e}_x,\pm\hat{e}_y,\pm\hat{e}_z$ and $\mathcal{Y}$ projective measurements given by the Bloch vectors $\hat{y}_1,...,\hat{y}_{\mathcal{Y}}$. Note that for convenience, we use unnormalized vectors $\vec y_j$ and write $\hat y_j=\vec y_j/\norm{\vec y_j}_2$. Hence the expectation value is given by  $E_{i,j}=\vec x_i\cdot \hat y_j$ and we construct the coefficients of the witness $W_{6,\mathcal{Y},2}$ via
\begin{equation}\label{eq:witness_from_bloch}
    c_{i,j}=\frac{\vec{x}_i \cdot \vec{y}_j}{2}.
\end{equation}
A classical model always exists if $d_C \ge \mathcal{X}$, yielding the algebraic maximum of the witness. For a witness of that form the algebraic maximum is given by
\begin{equation}
    W_{6,\mathcal{Y},2} \le \sum_{j} \norm{\vec{y}_j}_1 = \sum_{j,\ell} \abs{\vec{y}_{j,\ell}}.
\end{equation}
Using the qubit setting, one reaches the qubit value
\begin{equation}
\label{eq:witness_qubit_value}
\begin{split}
    Q &= \sum_{x,y} c_{x,y} E_{x,y}\\
    & = \frac{1}{2}\sum_{i,j} (\vec{x}_i\cdot\vec{y}_j)\left(\vec{x}_i\cdot\frac{\vec{y}_j}{\norm{\vec{y}}_2}\right)\\
    &= \sum_j \norm{\vec{y}_j}_2.
\end{split}\end{equation}

We consider first the case of $\mathcal{Y}=5$ projective measurements given by the  (unnormalized) vectors
\begin{equation}
    \begin{tabular}{c|>{\centering\arraybackslash}p{0.4cm}>{\centering\arraybackslash}p{0.4cm}>{\centering\arraybackslash}p{0.4cm}>{\centering\arraybackslash}p{0.4cm}>{\centering\arraybackslash}p{0.4cm}}
$i$                    & 1 & 2 & 3 & 4 & 5\\ \hline 
\multirow{3}{*}{$\vec{y}_i$}
& 1  & 1  & 1  & 1  & 1\\
& -1 & -1 & 0  & 2  & 2\\
& 1  & 0  & -1 & -1 & 2
\end{tabular}
\end{equation}
Then, the witness $W_{6,5,2}$ as described in Eq.~\eqref{eq:witness_from_bloch} has the classical bounds
\begin{equation}
    W_{6,5,2} \overset{\text{C}2}{\le} 8 \overset{\text{C}3}{\le} 10 \overset{\text{C}4}{\le} 12 \overset{\text{C}5}{\le} 14 \overset{\text{C}6}{\le} 16,
\end{equation}
where the bounds are for $d_C=2$ (C2), $d_C=3$ (C3), etc. These bounds have been found by exhaustive search over all classical strategies, see Appendix~\ref{appendix:C}.
The qubit value according to Eq.~\eqref{eq:witness_qubit_value} is given by
\begin{align}
    Q = 3+2\sqrt{2}+\sqrt{3}+\sqrt{6} \approx 10.01,
\end{align}
which certifies that a qubit is not trit-simulable albeit with a very small relative violation $Q/C_3-1\approx 0.001$.

In a setting with $\mathcal{Y}=10$, we consider the projective measurements given by the  (unnormalized) vectors
\begin{equation}
\begin{tabular}{c|>{\centering\arraybackslash}p{0.4cm}>{\centering\arraybackslash}p{0.4cm}>{\centering\arraybackslash}p{0.4cm}>{\centering\arraybackslash}p{0.4cm}>{\centering\arraybackslash}p{0.4cm}>{\centering\arraybackslash}p{0.4cm}>{\centering\arraybackslash}p{0.4cm}>{\centering\arraybackslash}p{0.4cm}>{\centering\arraybackslash}p{0.4cm}>{\centering\arraybackslash}p{0.4cm}}
$i$                    & 1 & 2 & 3 & 4 & 5 & 6 & 7 & 8 & 9 & 10 \\ \hline 
\multirow{3}{*}{$\vec{y}_i$}
& 1  & -1 & 1  & 1  & 1  & 1  & 0  & 1  & 1  & 0\\
& 1  & 1  & -1 & 1  & 1  & 0 & 1  & -1  & 0  & 1\\
& 1  & 1  & 1  & -1 & 0  & 1  & 1  & 0 & -1  & -1
\end{tabular}
\end{equation}
Then, the witness $W_{6,10,2}$ as described in Eq.~\eqref{eq:witness_from_bloch} has the classical bounds
\begin{equation}
    W_{6,10,2} \overset{\text{C}2}{\le} 12 \overset{\text{C}3}{\le} 15 \overset{\text{C}4}{\le} 18 \overset{\text{C}5}{\le} 21 \overset{\text{C}6}{\le} 24.
\end{equation}
The qubit setting achieves
\begin{equation}
    Q=6\sqrt{2}+4\sqrt{3} \approx 15.41,    
\end{equation}
and $Q/C_3 -1 \approx0.027$. This setting is highly symmetric, as in the Bloch representation, the state vectors point to the $6$ face-centered points of a cube and the measurement vectors point to the $12$ edge-centered points and $8$ vertices of the same cube.

\subsection{Qutrit}
\label{subsec:witness_qutrit}
In this section we will show that the classical simulation of the (real) qutrit requires at least a five letter alphabet, that is, $d_C\ge5$. Here the term \emph{real} refers to a quantum system where all states and measurements are restricted to real-valued coefficients. To this end, we will examine the Yu-Oh setting \cite{YuOh}, which has previously been applied in the context of Kochen-Specker contextuality \cite{Kochen_Specker}. In Section \ref{sec:lo_bounds_graph_coloring} we  discuss the connection to graph coloring and we show that the simulation cost $d_C$ for a given set of input states is lower bounded by the chromatic number of the corresponding orthogonality graph. This consideration yields for the Yu-Oh setting already $d_C\ge 4$. Here we improve this lower bound to $d_C\ge5$.

\begin{res}
    A classical simulation of the (real) qutrit requires at least a five letter alphabet, that is, $d_C\ge5$.
\end{res}

The $13$ Yu-Oh states $\ket{\phi_i} = \sum_k z_{i,k} \ket{k}, i\in[13]$ are given by the following state vectors, up to normalization.
\begin{equation}
\label{eq:yu_oh_states}
\begin{tabular}{c|>{\centering\arraybackslash}p{0.4cm}>{\centering\arraybackslash}p{0.4cm}>{\centering\arraybackslash}p{0.4cm}>{\centering\arraybackslash}p{0.4cm}>{\centering\arraybackslash}p{0.4cm}>{\centering\arraybackslash}p{0.4cm}>{\centering\arraybackslash}p{0.4cm}>{\centering\arraybackslash}p{0.4cm}>{\centering\arraybackslash}p{0.4cm}>{\centering\arraybackslash}p{0.4cm}>{\centering\arraybackslash}p{0.4cm}>{\centering\arraybackslash}p{0.4cm}>{\centering\arraybackslash}p{0.4cm}}
$i$                    & 1 & 2 & 3 & 4 & 5 & 6 & 7 & 8 & 9 & 10 & 11 & 12 & 13 \\ \hline 
\multirow{3}{*}{$\ket{\phi_i}$}
& 1  & -1 & 1  & 1  & 1  & 1  & 0  & 1  & 1  & 0   & 1   & 0   & 0   \\
& 1  & 1  & -1 & 1  & 1  & 0 & 1  & -1  & 0  & 1   & 0   & 1   & 0   \\
& 1  & 1  & 1  & -1 & 0  & 1  & 1  & 0 & -1  & -1  & 0   & 0   & 1  
\end{tabular}
\end{equation}
We construct the $13$ binary measurements $M_{\pm,j}$ also by using the vectors $\ket{\phi_i}$, that is, $M_{+\vert j} = \ket{\phi_j}\bra{\phi_j}$ and $M_{-\vert j} = \openone - \ket{\phi_j}\bra{\phi_j}$. The improved linear program (see Appendix \ref{appendix:A}) allows us to certify that the setting is simulable with $d_C=5$, but not with $d_C=4$. Note that the $(10,4,2)$ setting with the $10$ states $\ket{\phi_i}$, $i\in[10]$ and $4$ measurements $M_{\pm\vert j}$, $j\in[4]$ is sufficient to disprove a classical model with $d_C=4$. The coefficients of the witness extracted from the solution of the dual problem, see Ref.~\cite{Renner_PandM}, are given by
\begin{equation}\label{eq:yuohwitness}
    C^\intercal = 
    \left(\begin{array}{cccccccccc}
    -2 & 3 & 3 & 3 & -3 & -3 & -3 & 2 & 2 & 2 \\
    3 & -2 & 3 & 3 & 2 & 2 & -3 & -3 & -3 & 2 \\
    3 & 3 & -2 & 3 & 2 & -3 & 2 & -3 & 2 & -3 \\
    3 & 3 & 3 & -2 & -3 & 2 & 2 & 2 & -3 & -3 \\
    \end{array}\right),
\end{equation}
where the $(i,j)$ element of $C$ corresponds to the coefficient $c_{i,j}$ of the witness, cf.\ Eq.~\eqref{eq:witness_correlators}. This witness for the $(10,4,2)$ scenario obeys the following classical bounds
\begin{equation}\begin{split}
    W_{10,4,2} 
    &\overset{\text{C}2}{\le} 44 \overset{\text{C}3}{\le} 62 \overset{\text{C}4}{\le} 68 \overset{\text{C}5}{\le} 76 \overset{\text{C}6}{\le} 80\\
    &\overset{\text{C}7}{\le} 88 \overset{\text{C}8}{\le} 92 \overset{\text{C}9}{\le} 98 \overset{\text{C}10}{\le} 104.
\end{split}\end{equation}
The qutrit setting achieves $Q = 72$ and $Q/C_4 -1 = \frac{1}{17} \approx 0.059$. Thus, the real qutrit does not admit a classical simulation model with $d_C=4$.

We remark that the witness that was constructed for the entanglement-assisted prepare-and-measure scenario in Ref.~\cite{Pauwels2022} is also of interest for the prepare-and-measure scenario without entanglement. It allows one to distinguish between qutrit behaviors and classical behaviors realizable with  $d_C=4$. This scenario constitutes a $(9,4,3)$ setting where the input states $\{\ket{\psi_i}\}_i$ correspond to the Hesse symmetric informationally complete POVM \cite{Dang_2013}, that is, up to normalization,
\begin{equation}
\begin{tabular}{c|>{\centering\arraybackslash}p{0.45cm}>{\centering\arraybackslash}p{0.45cm}>{\centering\arraybackslash}p{0.45cm}>{\centering\arraybackslash}p{0.45cm}>{\centering\arraybackslash}p{0.45cm}>{\centering\arraybackslash}p{0.45cm}>{\centering\arraybackslash}p{0.45cm}>{\centering\arraybackslash}p{0.45cm}>{\centering\arraybackslash}p{0.45cm}>{\centering\arraybackslash}p{0.45cm}>{\centering\arraybackslash}p{0.45cm}>{\centering\arraybackslash}p{0.45cm}>{\centering\arraybackslash}p{0.45cm}}
$i$                    & 1 & 2 & 3 & 4 & 5 & 6 & 7 & 8 & 9\\ \hline 
\multirow{3}{*}{$\ket{\psi_i}$}
& 0 & -1 & 1  & 0       & $-\omega$  & 1       & 0          & $-\omega^2$  & 1\\
& 1 & 0 & -1 & 1        & 0        & $-\omega$ & 1          & 0          & $-\omega^2$\\
& -1 & 1 & 0  & $-\omega$ & 1        & 0       & $-\omega^2$  & 1          & 0 
\end{tabular}\quad,
\end{equation}
with $\omega = e^{2\pi i/3}$. Bob performs measurements in the four mutually unbiased bases,
\begin{equation}\label{eq:mub3}
    \begin{pmatrix} 1 & 0 & 0 \\ 0 & 1 & 0 \\ 0 & 0 & 1 \end{pmatrix},
    \begin{pmatrix} 1 & 1 & 1 \\ 1 & \omega & \omega^2 \\ 1 & \omega^2 & \omega \end{pmatrix},
    \begin{pmatrix} 1 & \omega & \omega \\ \omega & 1 & \omega \\ \omega & \omega & 1 \end{pmatrix},
    \begin{pmatrix} 1 & \omega^2 & \omega^2 \\ \omega^2 & 1 & \omega^2 \\ \omega^2 & \omega^2 & 1 \end{pmatrix}.
\end{equation}
The witness $W_{9,4,3}$ constructed in Ref.~\cite{Pauwels2022} has the coefficients $c_{x,y}^b = 4p_Q(b\vert x,y)-1$. We find the following classical bounds,
\begin{equation}
    W_{9,4,3} \overset{\text{C}2}{\le} 30 \overset{\text{C}3}{\le} 33 \overset{\text{C}4}{\le} 35 \overset{\text{C}5}{\le} 36.
\end{equation}
For the given, complex-valued, qutrit states and measurements, we have $Q=36$, that is, the witness attains its algebraic maximum and $Q/C_4 -1\approx 0.029$, confirming $d_C\ge 5$ for this setting \cite{Havlicek2020}.

\section{\label{sec:lo_bounds_graph_coloring}
Restricting classical models via state discrimination}
In this section, we  derive lower bounds on the worst-case communication from scenarios where the set of states and measurements is finite. First, we show that any classical model for a given set of states is restricted by the distinguishability of the states. We then derive a necessary condition for the existence of a classical simulation protocol. From this, we derive a linear program for certifying lower bounds on the simulation cost.
Furthermore, we relate the task of finding lower bounds to graph theory and we find that a lower bound on the simulation cost can be obtained from the chromatic number of the exclusivity graphs of quantum states.

\subsection{Restricting the classical model}
Every classical simulation protocol with a finite set of states has to satisfy the following property \cite{Montina2011}:
\begin{restatable}[]{lem}{corelemma}
\label{lem:core_lemma}
For all $c$ and for all $\lambda$ with $\pi(\lambda)>0$, the set $S_c(\lambda)= \{\rho_x \mid p_A(c|x,\lambda)>0\}$ does not contain any pair of orthogonal states.
\end{restatable}
The intuition for this property of a simulation protocol is as follows. In every round, Bob's knowledge about the state $\rho_x$ is limited to the shared variable $\lambda$ and the classical message $c$ he received from Alice. Thus, Bob can characterize the set $S_c(\lambda)$ and infer that $\rho_x\in S_c(\lambda)$. Suppose now that $S_c(\lambda)$ contains a pair of orthogonal states $\rho_+,\rho_-$ and Bob is asked to perform a measurement that perfectly discriminates between these two states. Then, he has no way to assign the correct outcome with certainty and the classical simulation model fails to reproduce the correct statistics.

This condition has to be fulfilled in order for the classical simulation to perfectly reproduce the statistics of orthogonal state discrimination. We now show that in general, the distinguishability of the input states strongly constrains the classical models. As discussed before, one can always assume that Alice acts deterministically by absorbing her local randomness into the shared randomness $\lambda$, cf.\ Eq.~\eqref{eq:prob_to_detstr_1} and Eq.~\eqref{eq:prob_to_detstr_2}. In case of a finite setting $(\mathcal{X},\mathcal{Y},\mathcal{B})$, the number of deterministic strategies is finite and the deterministic strategies are enumerated by $\lambda$. As a consequence, we can cast the integral as a finite sum and the classical model reads
\begin{equation}
\label{eq:cl_model_finite_setting}
    p_C(b\vert x,y) = \sum_{\lambda} \pi(\lambda)\sum_{c=1}^{d_C} D_A(c\vert x,\lambda)  p_B(b\vert c,y,\lambda),
\end{equation}
where $D_A(c\vert x,\lambda)\in\set{0,1}$ is the deterministic strategy of Alice with $\sum_c D_A(c|x,\lambda)=1$, and $\lvert\Lambda\rvert = d_C^{\mathcal{X}}$, see also Appendix~\ref{appendix:A}.

We define the set of deterministic strategies, labeled by $\lambda$, that assign the same message to a given pair of inputs $x$ and $x^\prime$, that is,
\begin{equation}
    \Lambda_{=}(x,x^\prime) = \set{\lambda\in\Lambda\vert D_A(c\vert x,\lambda) = D_A(c\vert x^\prime,\lambda) \forall c}.
\end{equation}
Accordingly, the set of deterministic strategies that assign different messages to the inputs $x$ and $x^\prime$ is given by $\Lambda_{\ne}(x,x^\prime) = \Lambda \setminus \Lambda_{=}(x,x^\prime)$.
Suppose now that the set of input states $\set{\rho_x}_{x=1}^\mathcal{X}$ contains at least one pair of orthogonal states. Then, from Lemma~\ref{lem:core_lemma} we see that in particular any deterministic strategy $D_A(c\vert x,\lambda)$ must not assign the same message this pair of orthogonal states, unless the corresponding weight $\pi(\lambda)$ is zero. Generalizing this result, we will now see that the contribution of strategies that assign different messages to a pair of states $\rho_x,\rho_{x'}$ is lower bounded by the trace distance or the distinguishability of the states. Here, $\left\lVert T \right\rVert_1 =\tr(\sqrt{T^\dag T})$ is the trace norm.
\begin{res}
\label{res:state_discrimination}
    Given a classical simulation protocol for a finite set of states $\mathcal{P}=\set{\rho_x}_{x=1}^{\mathcal{X}}$ and a set of measurements $\mathcal{M}$ that includes the optimal measurements for state discrimination for every pair of states, it holds that
    \begin{equation}
    \label{eq:restricted_cl_model}
        \sum_{\lambda\in\Lambda_{\ne}(x,x^\prime)} \pi(\lambda) \ge \frac{1}{2} \left\lVert \rho_x - \rho_{x^\prime} \right\rVert_1 \quad\forall x,x'\in [\mathcal X].
    \end{equation}
\end{res}

The intuition for this property of a simulation protocol is as follows. If Alice and Bob have to be able to perform optimal state discrimination for a pair of input states, the weights $\pi(\lambda)$ of the deterministic strategies that allow them to discriminate between this pair of states have to be sufficiently large. On the other hand, if the weights were lower, they would not be able to reach the Helstrom bound. In this sense, the Helstrom measurements are particularly difficult to simulate by a classical protocol.

\begin{proof}[Proof of Result~\ref{res:state_discrimination}]
    Consider the task of minimum error quantum state discrimination where a source prepares either $\rho_x$ or $\rho_{x^\prime}$ with equal probability and we try to infer the label $x,x^\prime$ by performing a binary measurement with effects $\{M_+,M_-\}$. The success probability is given by
    \begin{equation}\begin{split}
        p_{\text{success}}
        &= \frac{1}{2} \left[\tr(\rho_x M_+)+\tr(\rho_{x^\prime} M_-) \right]\\
        &= \frac{1}{2}+\frac{1}{2}\tr\left[(\rho_x-\rho_{x^\prime})M_+\right].
    \end{split}\end{equation}
    The success probability is maximized by performing the Helstrom measurement \cite{Helstrom1967} for the states $\rho_x,\rho_{x^\prime}$, which we denote by $M_{\pm\vert y}$. It is given by the projectors on the positive and negative eigenspace of $\rho_x - \rho_{x^\prime}$, and the resulting success probability is given by
    \begin{equation}
        p_{\text{success}} = \frac{1}{2} + \frac{1}{4} \left\lVert \rho_x - \rho_{x^\prime} \right\rVert_1.
    \end{equation}
    For the classical simulation, we have

    \begin{align}
        \frac{1}{2}\left\lVert \rho_x - \rho_{x^\prime} \right\rVert_1
        &= \tr[(\rho_x - \rho_{x^\prime})M_{+\vert y}]\nonumber\\
        &= \sum_{\lambda\in\Lambda} \pi(\lambda)\sum_{c=1}^{d_C} p_B(+\vert c,y,\lambda)\nonumber\\
        &\quad\quad\quad\quad\quad
        (D_A(c\vert x,\lambda) - D_A(c\vert x^\prime,\lambda))\nonumber\\
        &= \sum_{\lambda\in\Lambda_{\ne}(x,x^\prime)} \pi(\lambda)\sum_{c=1}^{d_C} p_B(+\vert c,y,\lambda)\nonumber\\
        &\quad\quad\quad\quad\quad
        (D_A(c\vert x,\lambda) - D_A(c\vert x^\prime,\lambda))\nonumber\\
        &\le \sum_{\lambda\in\Lambda_{\ne}(x,x^\prime)} \pi(\lambda),
    \end{align}
    where the last inequality follows by omitting the negative terms and using that $\sum_{c=1}^{d_C} p_B(+\vert c,y,\lambda)D_A(c\vert x,\lambda)\le 1$. As this holds for every pair $x,x^\prime$, the statement is proven.
\end{proof}

Crucially, Result~\ref{res:state_discrimination} provides a necessary condition for the existence of a classical model and thus can be applied to certify lower bounds on the simulation cost. It can be formulated as the following linear program:
\begin{align}
\label{eq:lp_new_LP_lowerbounds}
        \text{given\quad}&\set{\rho_x}_{x=1}^\mathcal{X}, d_C\\
        \text{find\quad}&\pi(\lambda), \quad \lambda\in\set{1,\ldots,d_C}^\mathcal{X}\nonumber\\
        \text{s.t.\quad}&\sum_{\lambda:\lambda_x\ne \lambda_{x'}} \pi(\lambda) \ge \frac{1}{2} \left\lVert \rho_x - \rho_{x'} \right\rVert_1 \quad\forall x,x',\nonumber\\
        &\pi(\lambda) \ge 0 \quad \forall\lambda,\nonumber\\
        &\sum_{\lambda} \pi(\lambda) = 1,\nonumber
\end{align}
where $\lambda_x$ corresponds to the message $c$ for which the corresponding deterministic strategy $D_A(c\vert x,\lambda) = 1$.

This linear program can certify that simulating a set of states $\set{\rho_x}_{x=1}^\mathcal{X}$ with a classical message of length $d_C$ is impossible, without treating the measurements explicitly because we implicitly assume that in particular, the classical model has to be capable of simulating the Helstrom measurements. As an important application: choosing two unbiased bases for the qutrit, see also Eq.~\eqref{eq:mub3}, the linear program is infeasible for $d_C=4$, providing the arguably simplest proof that for the qutrit, $d_C\ge 5$. 
The linear program involves fewer variables and constraints than the linear program for simulating the behaviors, see Eq.~\eqref{eq:lp_feasibility} in Appendix~\ref{appendix:A}. The number of variables is $N_\mathrm{var} = \abs{\Lambda}$ and the number of constraints is given by $N_\mathrm{con} = \abs{\Lambda} + \mathcal{X}(\mathcal{X}-1)/2+1$.

Naively, we have $\abs{\Lambda} = d_C^\mathcal{X}$. However, since the constraints depend only on differences between labels and not the specific labels themselves, they are invariant under relabeling. Therefore, it is sufficient to keep only one representative from each equivalence class of labelings. Furthermore, it suffices to only consider $\lambda$ that use every message $c\in[d_C]$ at least once, as we can always redistribute the weight of assignments that omit one or more labels to those that use all labels. Consequently, $\abs\Lambda$ is given by the Stirling number of the second kind $\stirling{\mathcal{X}}{d_C}$ which counts the number of ways to partition a set of $\mathcal{X}$ elements into $d_C$ non-empty subsets.

In addition, it is sufficient to consider the deterministic strategies that correspond to proper colorings of the orthogonality graph $G$ for the input states $\set{\rho_x}_{x=1}^\mathcal{X}$, since the constraints of the problem enforce the proper coloring. In the next section we will draw the connection to graph theory and graph coloring of orthogonality graphs. Furthermore, in Appendix~\ref{appendix:A}, we discuss the reduction of deterministic strategies in the linear program in more detail, but we will briefly present the main idea here. To this end, note that Eq.~\eqref{eq:restricted_cl_model} implies that
\begin{equation}
\label{eq:excluded_strategies}
    \sum_{\Lambda_=(x,x^\prime)} \pi(\lambda) = 0 \quad \forall x,x^\prime\in[\mathcal X] \;\text{s.t}\; \tr(\rho_x\rho_{x^\prime}) = 0,
\end{equation}
since $\pi(\lambda)\ge 0$ and $\sum_{\lambda\in\Lambda} \pi(\lambda) = 1$.
This is particularly useful as it allows us to reduce the number of variables in the linear program, as we can exclude the $\lambda\in\Lambda$ for which $\pi(\lambda)=0$ beforehand (see also Appendix \ref{appendix:A}).

\subsection{Formulation using Graph theory}
\label{subsec:graph_theory}
We now use graph theory to obtain lower bounds on the communication cost using Lemma~\ref{lem:core_lemma}. To this end, we consider a finite prepare-and-measure setting where the set of inputs $[\mathcal X]$ corresponds to vectors $\ket{\psi_x}$ in a $d_Q$-dimensional Hilbert space. We define the orthogonality graph $G$, where the  vertices are given by $\ket{\psi_x}$ and a pair of vertices is connected if the states are orthogonal $\braket{\psi_x|\psi_{x'}}=0$. The chromatic number $\chi(G)$ denotes the minimum number of colors required for a proper coloring of $G$, that is, a vertex coloring such that adjacent vertices do not share the same color. The chromatic polynomial $P(G,k)$ is a polynomial in $k$ such that its value is the number of proper graph colorings of $G$ using $k$ colors. Thus it is clear that, $\chi(G) = \min\set{k\in\mathbb{N} \vert P(G,k)>0}$.

We can think of Alice's deterministic strategies $D_A(c\vert x,\lambda)$ as colorings of $G$, where for a given $\lambda$, we assign to each vertex $\ket{\psi_x}$ the message (color) $c$ for which $D_A(c\vert x,\lambda)=1$. From Eq.~\eqref{eq:excluded_strategies}, it follows that all colorings which are not proper colorings of $G$ result in $\pi(\lambda)=0$. In particular, when $d_C < \chi(G)$, then there exists no proper coloring of $G$ with $d_C$ colors and hence the classical simulation fails.
\begin{lem}
\label{obs:lower_bound_chrom_number}
    The simulation cost $d_C$ for a quantum system of dimension $d_Q$ is lower bounded by the chromatic number $\chi(G)$ of any orthogonality graph $G$ that has a $d_Q$-dimensional representation. That is,
    \begin{equation}
    \label{eq:lower_bound_from_chromatic_number}
        d_C\ge\max_{G}\chi(G),
    \end{equation}
    where the maximum is over all orthogonality graphs of quantum states in dimension $d_Q$.
\end{lem}

This Lemma directly implies the trivial lower bound $d_C\ge d_Q$, as $\chi(K_{d_Q}) = d_Q$ where $K_{d_Q}$ denotes the complete graph with $d_Q$ vertices, which in dimension $d_Q$ is the orthogonality graph associated to a complete basis of quantum states.

Importantly, there exists exist orthogonality graphs with a chromatic number that is larger than the dimension of their representation. However, we have to go beyond qubits to find such graphs, since every qubit state has only one orthogonal state and thus, every qubit orthogonality graph is $2$-colorable. For qutrits, we can consider the Yu-Oh graph $G_{\text{YO}}$ with $\chi(G_{\text{YO}}) = 4$ \cite{YuOh,Mancinska2016}. The graph has a representation in terms of $13$ real-valued states in dimension $d_Q=3$, see Eq.~\eqref{eq:yu_oh_states}. As a consequence, we have that $d_C\ge d_Q+1$ holds for any simulation protocol with $d_Q\ge 3$, since if the dimension of the representation is increased by one, we can add a vertex that is orthogonal to all previous vertices, increasing the chromatic number by one. Furthermore, for increasing quantum dimension $d_Q$, one can construct families of graphs such that the chromatic number scales exponentially in $d_Q$, for example using Hadamard graphs. Lower bounds on the chromatic number of Hadamard graphs \cite{Frankl_HadamardGraphs,Ihringer2019TheIN} imply that for $d_Q = 4k$, where $k$ is a prime power, it holds that
\begin{equation}
    d_C > \left(\frac{27}{16}\right)^{d_Q/4} > 1.139^{d_Q},
\end{equation}
see Appendix~\ref{appendix:D}. Note that this bound can be surpassed by other methods, see Refs.~\cite{Montina2011,Havlicek2020}. Furthermore, Ref.~\cite{Winter_ChromaticNumber} provides an upper bound on the chromatic number of all orthogonality graphs with a quantum representation in dimension $d_Q$, that is, 
\begin{equation}
\label{eq:upper_bound_chrom_number}
    \max_{G}\chi(G) \le \left( 1+2\sqrt{2} \right)^{2 d_Q}<14.657^{d_Q}.
\end{equation}
We want to stress that while Eq.~\eqref{eq:upper_bound_chrom_number} limits the lower bounds on $d_C$ that can be obtained from the chromatic number of orthogonality graphs, the lower bounds that can be obtained from our linear program in Eqs.~\eqref{eq:lp_new_LP_lowerbounds} can exceed the chromatic number. More precisely, the linear program is infeasible when $d_C < \chi(G)$ but $d_C \ge \chi(G)$ does not imply feasibility: the previous example of two mutually unbiased bases in $d_Q=3$ has chromatic number $\chi(G)=3$ but requires $d_C=5$.

\section{\label{sec:discussion}Discussion}
In this work we have studied classical simulations of single-shot prepare-and-measure scenarios. We first discussed the two main notions of communication cost, namely the average communication cost and the worst-case communication cost. Motivated by Toner and Bacon~\cite{TonerBacon}, the worst-case communication cost has been studied in great detail; we here demonstrated that the average cost can be significantly lower than the worst-case cost. For this we utilized the optimal worst-case protocol for the qubit to give bounds on the average cost, yielding an explicit simulation protocol requiring \unit[1.89]{bits}. For higher quantum dimensions, this gap may become even more significant.

A worst-case communication cost $d_C>d_Q$ can be understood as a quantum advantage where communicating a quantum system of dimension $d_Q$ is significantly more efficient than communicating a classical system of dimension $d_C$. Identifying minimal scenarios providing a quantum advantage is fundamental to our understanding of quantum systems and their possible technological applications. We studied minimal scenarios for such an advantage and provided corresponding witnesses, that is, correlation inequalities where achieving the qubit (qutrit) value requires a classical system of dimension $d_C=4$ ($d_C=5$). The smallest settings that we identified, for example, uses 6 (10) state preparations and 5 (4) two-outcome measurements.

Lower bounds on the worst-case communication cost are typically constructed by providing a (finite) set of states and measurements and showing that the corresponding quantum behavior cannot be simulated classically. We find that using the measurements for optimal pairwise state discrimination already poses strong constraints on the classical dimension. The resulting linear program can be used to rule out the simulability of a set of states without specifying the measurements, thereby also significantly reducing the size of the linear program. As an application, we find that as little as 6 qutrit states suffice to achieve a simulation cost of $d_C\ge 5$.

The most important open question is the classical simulation of a qutrit, in particular, whether the classical dimension can be finite, $d_C< \infty$. We identify a different open simulation problem: the real qubit. Since it is a subset of the qubit, the communication cost cannot be higher than for the qubit, but it is an interesting (and surprisingly intricate) question whether a strictly more efficient simulation with $d_C=3$ is possible. We show that if this is the case, then the resulting protocol cannot be of a simple form. Developing methods to answer this question could be a key ingredient to approach the qutrit case, in particular, by studying low-dimensional sections of the qutrit.

\begin{acknowledgments}
We thank Carlos de Gois, Giulio Gasbarri, Martin J.\ Renner, Pauli Jokinen, Armin Tavakoli and Roope Uola for discussions. We acknowledge the usage of the OMNI cluster of Uni Siegen for computations.
This work was supported by the
Deutsche Forschungsgemeinschaft (DFG, German Research Foundation, project numbers 447948357 and 440958198),
the Sino-German Center for Research Promotion (Project M-0294),
the German Ministry of Education and Research (Project QuKuK, BMBF Grant No. 16KIS1618K), and
the Wallenberg Initiative on Networks and Quantum Information (WINQ).
\end{acknowledgments}

\twocolumngrid
\bibliography{bibliography}

\onecolumngrid

\appendix

\section{\label{appendix:B}The average communication cost of the qubit protocol}
The average communication cost of the protocol is given by
\begin{equation}
    \bar{\mathcal{C}} = \int_0^\pi\text{d}\alpha \frac{\sin\alpha}{2} L(\alpha),
\end{equation}
where $L(\alpha)$ denotes the average codeword length of the encoding and $\frac{1}{2}\sin\alpha=\frac1{4\pi}(\int \dd\phi) \sin \alpha$ is the distribution of the angle $\alpha$ enclosed by a pair of vectors that are drawn independently from a uniform distribution on the unit sphere.

Since our qubit simulation protocol involves a piecewise encoding, the integral is defined piecewise and we have
\begin{align}
            \bar{\mathcal{C}} =
            &\int_{0}^{\pi\log_2{\left(9/8\right)}}
            {\text{d}\alpha\left(\frac{3}{2}+\frac{3\alpha}{2\pi}\right)\frac{\sin{\alpha}}{2}}
            +
            \int_{\pi\log_2{\left(9/8\right)}}^{\pi\log_2{\left(4/3\right)}}{\text{d}\alpha\left(\log_2{(3)}+\frac{\alpha}{\pi}\right)\frac{\sin{\alpha}}{2}}
            +
            \int_{\pi\log_2{(4/3)}}^{\pi\log_2{\left(3/2\right)}}
            {\text{d}\alpha\ 2\ \frac{\sin{\alpha}}{2}}\nonumber\\
            +
            &\int_{\pi\log_2{\left(3/2\right)}}^{\pi\log_2{\left(16/9\right)}}
            {\text{d}\alpha\left(\log_2{(3)} + 1 -\frac{\alpha}{\pi}\right)\frac{\sin{\alpha}}{2}}
            +
            \int_{\pi\log_2{\left(16/9\right)}}^{\pi}
            {\text{d}\alpha\left(3-\frac{3\alpha}{2\pi}\right)\frac{\sin{\alpha}}{2}}.
\end{align}
Due to symmetry, the integral can be simplified to
\begin{align}
            \bar{\mathcal{C}}
            &= \int_{0}^{\pi\log_2{\left(9/8\right)}}{\text{d}\alpha\left(\frac{3}{2}+\frac{3\alpha}{2\pi}\right)\sin{\alpha}}
            +
            \int_{\pi\log_2{\left(9/8\right)}}^{\pi\log_2{\left(4/3\right)}}{\text{d}\alpha\left(\log_2{(3)}+\frac{\alpha}{\pi}\right)\sin{\alpha}}
            +
            2\int_{\pi\log_2{(4/3)}}^{\frac{\pi}{2}}{\text{d}\alpha\sin{\alpha}},
        \end{align}
which evaluates to 
\begin{align}
    \bar{\mathcal{C}}
    &= \frac{3}{2} - \frac{\sin\left(\pi\log_2 3\right)}{\pi}- \frac{\sin\left(2\pi\log_2 3\right)}{2\pi}\\
    &\approx \unit[1.89]{bits}.
\end{align}

\section{\label{appendix:C}Classical bounds}
Given a classical dimension witness
\begin{equation}
    \sum_{b,x,y} c_{x,y}^b\  p_C(b|x,y) \le C_d,
\end{equation}
with $c_{x,y}^b\in\mathbb{R}$.
Since the set of classical behaviors forms a convex set and we are maximizing over a linear functional, the maximum is always attained at an extremal point and the classical bound $C_d$  is given by
\begin{align}
    C_d\
    &= \max_\lambda\left[\sum_{x,c,y,b}c_{x,y}^b\ D_A(c|x,\lambda)D_B(b|c,y,\lambda)\right]\\
    &= \max_\lambda\left[\sum_{c,y,b}D_B(b|c,y,\lambda)\sum_xc_{x,y}^b\ D_A(c|x,\lambda)\right].
\end{align}
In order to maximize this expression, Bob has to output the $b$ that maximizes $\sum_xc_{x,y}^b\ D_A(c|x,\lambda)$, see Refs.~\cite{Ara_jo_2020,Renner_PandM}. Therefore, the classical bound is given by
\begin{equation}
    C_d = \max_\lambda\left[\sum_{c,y}\max_b\sum_xc_{x,y}^b\ D_A(c|x,\lambda)\right].
\end{equation}

\section{\label{appendix:D}Exponential lower bounds from Hadamard graphs}
Let $d_Q$ be a natural number and let $V(d_Q) = \set{\pm1}^{d_Q}$ be the set of vectors that constitute the vertices of the graph. A pair of vertices is connected if the corresponding vectors are orthogonal, that is, the set of edges is given by $E = \set{(u,v)\vert\langle u \vert v \rangle=0,u\in V,v\in V}$. The graphs $H(d_Q) = (V(d_Q),E)$ are Hadamard graphs \cite{Ito1985}. Importantly, every Hadamard graph corresponds to an orthogonality graph of $d_Q$-dimensional quantum states. For each $u\in V(d_Q)$, we consider the quantum state
\begin{equation}
    \ket{\psi(u)} = \frac{1}{\sqrt{N}} \sum_{i=1}^{d_Q} u_i\ket{i}.
\end{equation}
Thus, by Lemma~\ref{obs:lower_bound_chrom_number}, the simulation cost for the $d_Q$-dimensional prepare-and-measure scenario is in particular lower bounded by the chromatic number of the corresponding Hadamard graph $H(d_Q)$, that is,
\begin{equation}
    d_C \ge \chi(H(d_Q)),
\end{equation}
where $\chi(G)$ denotes the chromatic number of the graph $G$. It has been shown that
\begin{equation}
    \alpha(G_{4k}) = 4\sum_{i=0}^{k-1} \binom{4k-1}{i} < \frac{4^{4k}}{3^{3k}},
\end{equation}
for $k=p^q$ where $p\ge3$ is prime and $q\ge1$ \cite{Frankl_HadamardGraphs} and for $k=2^r$ with $r\ge0$ \cite{Ihringer2019TheIN}. Here $\alpha(G)$ denotes the independence number of $G$, that is, the maximal size of a subset that does not contain adjacent vertices. Since $\chi(G)\ge\frac{\vert V \vert}{\alpha(G)}$, we have
\begin{equation}
    \chi(G_{4k}) \ge \frac{2^{4k}}{4\sum_{i=0}^{k-1} \binom{4k-1}{i}} > \frac{2^{4k}3^{3k}}{4^{4k}} = \left(\frac{27}{16}\right)^k.
\end{equation}
Consequently, for $d_Q=4k$,
\begin{equation}
    d_C(d_Q) \ge \chi(G_{4k}) \ge \frac{2^{4k}}{4\sum_{i=0}^{k-1} \binom{4k-1}{i}} > \left(\frac{27}{16}\right)^{d_Q/4},
\end{equation}
when $k$ is a prime power.

\section{\label{appendix:A}Improved formulation of the linear program}
Here, we will introduce the linear program that we employed for our work. The task of deciding whether a target \emph{behavior} $\set{p(b\vert x,y)}_{b,x,y}$ admits a classical model with a classical message of length $d_C$ can be cast as a linear program \cite{Gallego_2010,de_Gois_2021,Renner_PandM,Tavakoli_SDP_review}. First, we will formulate the problem, then we will describe an improvement to the scaling of the problem that allows us to solve instances that were previously intractable.


\subsection{Formulation of the primal problem}
Consider a finite prepare-and-measure scenario where Alice has $\mathcal{X}$ inputs and Bob has $\mathcal{Y}$ inputs and $\mathcal{B}$ outputs. The input for Alice is denoted by $x$, the input for Bob is denoted by $y$ and the output Bob produces is denoted by $b$. The resulting behavior $\set{p(b\vert x,y)}_{b,x,y}$, is given by Born's rule. The task of deciding whether a given behavior admits a classical simulation model with an alphabet length $d_C$ can be phrased as follows:
\begin{align}
        \text{given\quad}&\{p(b|x,y)\}_{b,x,y},d_C\\
        \text{find\quad}&\pi(\lambda),p_A(c|x,\lambda),p_B(b|c,y,\lambda)\nonumber\\
        \text{s.t.\quad}&p(b|x,y) = \int_\Lambda \text{d}\lambda \sum_{c=1}^{d_C} \pi(\lambda)\ p_A(c|x,\lambda)\ p_B(b|c,y,\lambda),\quad\forall b,x,y\nonumber\\
        &\pi(\lambda) \ge 0,\quad\forall \lambda\nonumber\\
        &p_A(c|x,\lambda) \ge 0,\quad\forall c,x,\lambda\nonumber\\
        &\sum_{c=1}^{d_C} p_A(c|x,\lambda) = 1,\quad\forall x,\lambda\nonumber\\
        &p_B(b|c,y,\lambda) \ge 0,\quad\forall b,c,y,\lambda\nonumber\\
        &\sum_{b=1}^{\mathcal{B}} p_B(b|c,y,\lambda) = 1,\quad\forall c,y,\lambda\nonumber
\end{align}
where $\int_\Lambda \text{d}\lambda\;\pi(\lambda)=1$ is fulfilled since $\sum_b p(b\vert x,y) = 1 \;\forall x,y$. When $d_C\ge \mathcal{X}$ or $d_C\ge \mathcal{Y} \mathcal{B}$, a simulation is trivially possible, as Alice can simply encode her state or a suitable measurement outcome for every measurement Bob can perform in the classical message. Note that the above formulation of the problem is not linear, as the objective function is cubic in the optimization variables. In order to rewrite the problem in linear form, it is key to notice that the distributions $\set{p_A(c\vert x,\lambda)}_{c,x,\lambda}$ form a polytope, where the $d_C^{\mathcal{X}}$ extremal points are given by the deterministic distributions $D_A(c\vert x,\lambda)\in\set{0,1}$. Likewise, the distributions $\set{p_B(b\vert c,y,\lambda)}_{b,c,y,\lambda}$ form a polytope, where the $\mathcal{B}^{d_C \mathcal{Y}}$ extremal points are given by the deterministic distributions $D_B(b\vert c,y,\lambda)$. We will refer to the extremal distributions as \emph{deterministic strategies}, as they can be interpreted as the strategies Alice and Bob pursue for a given $\lambda$. It is often useful to think about $\lambda$ as a label that fixes the deterministic strategy $D_A(c\vert x,\lambda)$. Alice is then given an input $x$ and produces the message $c$. Each base-$d_C$ representation of the numbers $\set{0,...,d_C^{\mathcal{X}}-1}$ can be identified with a deterministic strategy $D_A(c\vert x,\lambda)$. In particular, the message $c$ that Alice sends for a given $x$ is given by the $x$-th digit of the base-$d_C$ representation of $\lambda\in\set{0,...,d_C^{\mathcal{X}}-1}$. The same considerations can be applied to the deterministic strategies $D_B(b\vert c,y,\lambda)$ of Bob.

As the total number of deterministic strategies is finite, Fine's theorem \cite{Fine} then allows us to cast the integral as a finite sum, and the formulation of the problem becomes linear. We will be adopting a more recent formulation of the problem that was described in Ref.~\cite{Renner_PandM}. Instead of rewriting the problem such that both Alice and Bob act deterministically, we rewrite the problem such that only one party acts deterministically and we apply a transformation to the problem. Specifically, we let Alice act deterministically and we apply the transformation $p_B^\prime(b\vert c,y,\lambda)\coloneqq\pi(\lambda)p_B(b\vert c,y,\lambda)$ and the problem reads \cite{Renner_PandM}:
\begin{align}
\label{eq:lp_feasibility}
        \text{given\quad}&\{p(b|x,y)\},d_C\\
        \text{find\quad}&\pi(\lambda),p_B^{\prime}(b|c,y,\lambda)\nonumber\\
        \text{s.t.\quad}&p(b|x,y) = \sum_{\lambda\in\Lambda} \sum_{c=1}^{d_C} p_B^{\prime}(b|c,y,\lambda) D_A(c|x,\lambda),\quad\forall b,x,y\nonumber\\
        &p_B^{\prime}(b|c,y,\lambda) \ge 0,\quad\forall b,c,y,\lambda\nonumber\\
        &\sum_{b=1}^{\mathcal{B}} p_B^{\prime}(b|c,y,\lambda) = \pi(\lambda),\quad\forall c,y,\lambda\nonumber
\end{align}
where $\lvert\Lambda\rvert = d_C^{\mathcal{X}}$. That is, in this formulation, the size of the problem only scales exponentially in the number of inputs for Alice. As noted by the authors of Ref.~\cite{Renner_PandM}, when $\mathcal{B}^{d_C \mathcal{Y}}>d_C^{\mathcal{X}}$,it might be more efficient to let Bob act deterministically and apply the transformation $p_A^\prime(c\vert x,\lambda) \coloneqq \pi(\lambda) p_A(c\vert x,\lambda)$.

The above linear program is formulated as a feasibility problem, which can readily be transformed into a robustness optimization problem. Instead of deciding whether a given behavior can be classically simulated with a $d_C$-dimensional system, we will instead find the maximum visibility parameter $\eta$ such that the probabilities $\eta\;p(b\vert x,y)+(1-\eta)\frac{1}{\mathcal{B}}$ admit a classical model. The term $\frac{1}{\mathcal{B}}$ represents the probability of Bob correctly guessing the measurement outcome.
The robustness formulation of the linear program reads as follows \cite{Renner_PandM}:
\begin{align}
        \text{given\quad}&\{p(b|x,y)\},d_C\label{eq:lp_robustness}\\
        \text{max.\quad}&\eta\nonumber\\
        \text{s.t.\quad}&\eta\;p(b|x,y) + (1-\eta) \frac{1}{\mathcal{B}} = \sum_{\lambda\in\Lambda} \sum_{c=1}^{d_C} p_B^{\prime}(b|c,y,\lambda) D_A(c|x,\lambda),\quad\forall b,x,y\nonumber\\
        &p_B^{\prime}(b|c,y,\lambda) \ge 0,\quad\forall b,c,y,\lambda\nonumber\\
        &\sum_{b=1}^{\mathcal{B}} p_B^{\prime}(b|c,y,\lambda) = \pi(\lambda),\quad\forall c,y,\lambda\nonumber
\end{align}
where $\lvert\Lambda\rvert = d_C^{\mathcal{X}}$. Here, $\eta\ge 1$ certifies the existence of a classical model, while $\eta<1$ certifies that the given behavior does not admit a simulation using a classical alphabet of length $d_C$. In an instance where $\eta<1$, the dual problem provides a hyperplane that separates all classical behaviors and certain nonclassical behaviors. In Ref.~\cite{Renner_PandM}, the explicit form of the dual of the above problem is derived.

\subsection{Improving the linear program}
We will now describe two improvements that can be applied to both the feasibility formulation from Eqs.\ \eqref{eq:lp_feasibility}, as well as the robustness formulation from Eqs.\ \eqref{eq:lp_robustness}. The first improvement depends on orthogonality relations within the set of quantum states $\mathcal{P}$ that gives rise to the target behavior $\set{p(b\vert x,y)}_{b,x,y}$. The second improvement is fully general and does not depend on the specific structure of the problem, but it only leads to an improvement that is constant with respect to the input size. 

Let us start by discussing the first improvement that depends on the orthogonality relations within $\mathcal{P}$. We will assume that for every pair of orthogonal states, the set of measurements includes a measurement for perfect state discrimination, as it makes the formulation of the improvement easier. The adaptations that need to be made when this requirement is not met, can also be directly inferred from our subsequent discussions. From Result~\ref{res:state_discrimination} and also from Lemma~\ref{lem:core_lemma}, it follows that a deterministic strategy $D_A(c\vert x,\lambda)$ that assigns the same message $c$ to a pair of orthogonal states does not contribute to the classical model, that is, the corresponding weight $\pi(\lambda)=0$. As a consequence, we have that given a pair of orthogonal states $\rho_x,\rho_{x^\prime}$, it holds that,
\begin{equation}
    \sum_{\lambda\in\Lambda_{=}(x,x^\prime)} \pi(\lambda) = 0,
\end{equation}
where $\Lambda_{=}(x,x^\prime) = \set{\lambda\in\Lambda\vert D_A(c\vert x,\lambda) = D_A(c\vert x^\prime,\lambda) \forall c}$ (cf.\ Result~\ref{res:state_discrimination}). Thus, just by the orthogonality relations of $\rho_x$ and $\rho_{x^\prime}$, we can restrict the number of variables to $\lvert\Lambda_{\ne}(x,x^\prime)\rvert<\lvert\Lambda\rvert = d_C^{\mathcal{X}}$. When taking into account all pairs of orthogonal states, we have that
\begin{equation}
    \pi(\lambda) = 0 \quad\forall \lambda\in\bigcup_{\substack{\rho_{x},\rho_{x^\prime}\in\mathcal{P},\\\tr(\rho_x \rho_{x^\prime})=0}} \Lambda_=(x,x^\prime).
\end{equation}
Thus, we can exclude the variables $\lambda$ such that $\pi(\lambda)=0$, and the corresponding constraints from the problem. In this way we enforce constraints that are included in the formulation of the problem by removing unnecessary variables from the formulation of the problem.

The improved formulation of the feasibility problem from Eq.~\eqref{eq:lp_feasibility} then reads as follows:
\begin{align}
        \text{given\quad}&\{p(b|x,y)\},d_C\\
        \text{find\quad}&\pi(\lambda),p_B^{\prime}(b|c,y,\lambda)\nonumber\\
        \text{s.t.\quad}&p(b|x,y) = \sum_{\lambda\in\Lambda(\mathcal{P})} \sum_{c=1}^{d_C} p_B^{\prime}(b|c,y,\lambda) D_A(c|x,\lambda),\quad\forall b,x,y\nonumber\\
        &p_B^{\prime}(b|c,y,\lambda) \ge 0,\quad\forall b,c,y \quad\text{and}\quad \forall\lambda\in\Lambda(\mathcal{P})\nonumber\\
        &\sum_{b=1}^{\mathcal{B}} p_B^{\prime}(b|c,y,\lambda) = \pi(\lambda),\quad\forall c,y \quad\text{and}\quad \forall\lambda\in\Lambda(\mathcal{P})\nonumber
\end{align}
where
\begin{equation}
    \Lambda(\mathcal{P}) = \Lambda \setminus \bigcup_{\substack{\rho_{x},\rho_{x^\prime}\in\mathcal{P},\\\tr(\rho_x \rho_{x^\prime})=0}} \Lambda_=(x,x^\prime),
\end{equation}
and $\Lambda_{=}(x,x^\prime) = \set{\lambda\in\Lambda\vert D_A(c\vert x,\lambda) = D_A(c\vert x^\prime,\lambda) \forall c}$. In words, the set $\Lambda(P)$ is given by all the $\lambda$ such that the corresponding deterministic strategy does not assign the same message to any pair of orthogonal states in $\mathcal{P}$. Every deterministic strategy $D_A(c\vert x,\lambda)$ corresponds to a coloring of the orthogonality graph induced by $\mathcal{P}$ (see Section \ref{sec:lo_bounds_graph_coloring}). In short, the orthogonality graph $G=(\mathcal{P},E)$ induced by the set of the states $\mathcal{P}$ consists of a set of vertices, given by $\mathcal{P}$, and a set of edges given by $E=\set{(\rho_x,\rho_{x^\prime})\vert \tr(\rho_x \rho_{x^\prime})=0,\rho_x\in\mathcal{P},\rho_{x^\prime}\in\mathcal{P}}$.

We can identify the $D_A(c\vert x,\lambda)$ with colorings of the vertices of $G(\mathcal{P},E)$, as for a given $\lambda$, they assign a message (color) $c$ to every state $\rho_x$. The set $\Lambda(\mathcal{P})$ corresponds to exactly those $\lambda$, for which $D_A(c\vert x,\lambda)$ constitutes a proper coloring of the orthogonality graph $G(\mathcal{P},E)$, that is, a coloring such that adjacent vertices do not have the same color. Thus, we have that
\begin{equation}
    \lvert\Lambda(\mathcal{P})\rvert = P(G,d_C),
\end{equation}
where $P(G,d_C)$ is the chromatic polynomial. It counts the number of proper $d_C$ colorings of the orthogonality graph $G=(\mathcal{P},E)$. As a consequence, in comparison to the original formulation, the number of variables is reduced by a factor of $ R =\frac{\lvert\Lambda(\mathcal{P})\rvert}{\lvert\Lambda\rvert}=\frac{P(G,d_C)}{d_C^{\mathcal{X}}}$ and similarly, the number of constraints reduces by approximately the same factor.

Whenever $\mathcal{P}$ does not contain any pair of orthogonal states, we have $R=1$, that is, the complexity of the problem is not reduced. On the other hand, when $\mathcal{P}$ contains at least one pair of orthogonal states, $R<1$. For example, when $\mathcal{P}$ consists of $n$ complete bases in $d_Q$ (and $d_C\ge d_Q$), we have
\begin{align}
    R 
    &= \frac{\lvert\Lambda(\mathcal{P})\rvert}{\lvert\Lambda\rvert}\\
    &= \left(\frac{P(K_{d_Q},d_C)}{d_C^{d_Q}}\right)^n\\
    &= \left(\frac{d_C!}{d_C^{d_Q}(d_C-d_Q)!}\right)^n,
\end{align}
where $K_q$ denotes the fully connected graph with $q$ vertices. For $d_C<d_Q$ we have $R=0$ which contradicts the existence of a classical model.

We will now discuss our second improvement to the linear program which can be applied in general, independently of the specific target behavior or the underlying set of quantum states. It is based on the observation that all deterministic strategies $D_A(c\vert x,\lambda)$ that are equivalent up to a relabeling of $c$ can be merged into a single deterministic strategy.

To this end, consider a pair of deterministic strategies $D_A(c\vert x,\lambda_1)$ and $D_A(c\vert x,\lambda_2)$ that are equivalent up to a relabeling. The simplest example for a pair of such strategies for $d_C=2$ ($c\in\set{1,2}$) is given by the strategy that assigns the message $c=1$ for all inputs and the strategy that assigns $c=2$ for all inputs.

It is key to notice that both these strategies are equivalent up to a local coin on Bob's side. That is, Bob receives the same information about $x$, whether they jointly draw $\lambda_1$ or $\lambda_2$, the only difference lies in a relabeling. However, this relabeling can also be performed by a local coin on Bob's side, that is, it can be accounted for by his responses $p_B^\prime(b\vert c,y,\lambda)$. Thus, for a set of strategies $D_A(c\vert x,\lambda)$ that are equivalent up to relabeling, it suffices to pick one of them and exclude the others from the classical model, effectively setting their weight $\pi(\lambda)$ to zero. Given the set of all deterministic strategies $\Lambda$ with $\lvert\Lambda\rvert = d_C^{\mathcal{X}}$, this improvement to the linear program allows us to reduce the number of deterministic strategies to
\begin{align}
    N_\lambda(d_C,\mathcal{X}) 
    &= \sum_{j=1}^{d_C} \stirling{\mathcal{X}}{j}\\
    &= \sum_{j=1}^{d_C} \frac{1}{j!}\sum_{k=1}^{j}(-1)^{j-k} \binom{j}{k} k^{\mathcal{X}}.  
\end{align}
Here, $\stirling{k}{l} $ denotes the Stirling number of the second kind which denotes the number of ways to partition a set of $k$ elements into $l$ non-empty subsets. For $\mathcal{X}\rightarrow\infty$, we have that $N_\lambda(d_C,\mathcal{X}) \sim \frac{d_C^{\mathcal{X}}}{d_C!}$. An efficient way to construct the valid strategies is given by constructing all words (colorings) that are in \emph{standard order}. A word of length $\mathcal{X}$ over an alphabet $[d_C]=\set{1,...,d_C}$ is in standard order if whenever a letter $s\;(2\le s\le d_C)$ appears, the letter $s-1$ has already appeared in the word \cite{Arndt}. Formally, the set of all words that are in standard order is given by
\begin{equation}
    W(d_C,\mathcal{X}) = \set{(s_1,s_2,...,s_{\mathcal{X}})\vert s_1=1,s_i\in[\min\set{\max\set{s_1,...,s_{i-1}}+1,d_C}]\;\forall i\in\set{2,...,\mathcal{X}}}.
\end{equation}

Both reductions can be combined such that the deterministic strategies that have to be considered are given by proper colorings of the orthogonality graph $G=(\mathcal{P},E)$ which are in standard order. They are a strict subset of all valid colorings of $G$ and thus, $R<1$ holds true for all graphs.

\end{document}